\documentclass[12pt,a4paper]{article}
\usepackage{orcidlink}

\usepackage[semicolon]{natbib}
\usepackage{xcolor}
\usepackage{graphicx} 
\usepackage{cancel}
\usepackage[normalem]{ulem}
\usepackage{subfigure,caption
}
\usepackage{amsmath}
\usepackage{amsfonts}
\usepackage{amssymb}
\usepackage{amsthm}

\usepackage{bbm,bm}
\usepackage{mysec}
\usepackage{enumitem}

\usepackage{hyperref}
\hypersetup{
    breaklinks=true, 
}

\usepackage{listings}
\newcommand{\myp}{\mbox{$\:\!$}}

\newcommand{\myn}{\mbox{$\;\!\!$}}
\newcommand{\mynn}{\mbox{$\:\!\!$}}

\newtheorem{theorem}{Theorem}

\newtheorem{lemma}
{Lemma}

{Proposition}

{Corollary}
{Conjecture}

\theoremstyle{definition}

\newtheorem{remark}{Remark}

\title{A citation index bridging Hirsch's~$h$ and Egghe's~$g$}
\author{\large
Ruheyan Nuermaimaiti,$^{\!\text{a,b}\,\orcidlink{https://orcid.org/0000-0002-5764-9949}}$ Leonid Bogachev$^{\myp\text{a}\,\orcidlink{0000-0002-2365-2621}}$ and Jochen Voss$^{\myp\text{a}\,\orcidlink{0000-0002-2323-3814}}$\\[.7pc]
{\small \tt \href{mailto:R.Nuermaimaiti@imperial.ac.uk}{R.Nuermaimaiti@imperial.ac.uk}; \href{mailto:L.V.Bogachev@leeds.ac.uk}{L.V.Bogachev@leeds.ac.uk}; \href{mailto:J.Voss@leeds.ac.uk}{J.Voss@leeds.ac.uk}}\\
{\small $^\text{a}$\myp Department of Statistics, School of Mathematics, University of Leeds, Leeds, LS2 9JT, UK}\\[-.2pc]
{\small 
$^\text{b}$\myp Department of  Mathematics, Faculty of Natural Sciences, 
Imperial College London,}\\[-.2pc] 
{\small South Kensington Campus, London, 
SW7 2AZ, UK}}
\date{}

\begin{document}
\maketitle

\begin{abstract}
\noindent
We propose a new citation index $\nu$ (``nu'') and show that it lies between the classical $h$-index and $g$-index. This idea is then generalized to a monotone parametric family $(\nu_\alpha)$ ($\alpha\ge 0$), whereby $h=\nu_0$ and  $\nu=\nu_1$, while the limiting value $\nu_\infty$ is expressed in terms of the maximum citation.


\bigskip\smallskip
\centerline{\bf Significance Statement}

\medskip\noindent
The widely used Hirsch’s $h$-index values productivity but overlooks how highly each paper is cited, while 
Egghe’s $g$-index emphasizes top-cited work but neglects lower-cited contributions. To address these imbalances, we propose the $\nu$-index, a synthetic metric that accounts for both highly and modestly cited publications and, therefore, offers a fairer and more balanced assessment of research impact.
\end{abstract}

\section{Introduction} \label{intro}

\subsection{\bf Background}
\citet{hirsch2005index} made a breakthrough in 
scientometrics by proposing for the first time a simple citation index (commonly referred to as \emph{$h$-index}), which had the advantage of aggregating the author's productivity on the basis of both the number of published papers and their quality measured by generated citations. Before that, only some extensive summary statistics were used, such as the mean number of citations per paper. Since then, the $h$-index has become a standard metric of authors' reputation and productivity, for instance routinely taken into consideration in academic appointments and promotions. 

Specifically, the $h$-index is defined as the maximum number $h$ of an author's papers, each cited at least $h$ times \citep{hirsch2005index}. 
Therefore, this index only takes into account the fact of a relatively ``high'' citation of a paper, but the actual number of citations of such a paper is effectively ignored. 

To remedy such censoring of larger citations, 
an alternative citation index (referred to as \emph{$g$-index}) was proposed by 
\citet{egghe2006improvement}, defined as the maximum number $g$ of an author's most cited papers, such that their total number of citations is at least $g^2\myn$. From this definition, it is easy to see that $h\le g$ \citep{egghe2006theory}.

These two (by now classical) indexes have attracted a lot of interest and generated ample research into their analytic properties and performance on real datasets, including their estimation in a variety of statistical models of count data (see, e.g., \citealp{egghe2006theory,er2006, egghe2008h,
{web-h}
}).
Furthermore, many modifications and alternative variants of the $h$ and $g$ indexes have been proposed, focusing on certain features of the citation profile (see, e.g.,  \citealp{alonso2009h, guns2009real, hirsch2019index,web-h}, and further references therein). 

\subsection{\bf New index and layout}
In the present work, we introduce a new citation index $\nu$ (``nu'') aiming to bridge the mathematical definitions of Hirsch's $h$ and 
Egghe's $g$. This idea was first coined in \cite{Nuer2023}. Namely, we start by observing that the $h$-index can be represented as a sum of certain indicator functions that censor papers to ensure a required minimum of citations. Building on this observation, our $\nu$-index essentially mimics the summative nature of the $h$-index but the new summation explicitly involves the 
numbers of citations of the top papers. 

We are then able to show that our $\nu$ is ``sandwiched'' between $h$ and $g$; more precisely, we prove the two-sided inequalities 
\begin{equation*}
h\le \nu\le g^*\!,
\end{equation*}
where $g^*\!$ denotes a modified (unconstrained) $g$-index, obtained if we are allowed to add fictitious zeros to the citation vector  \citep{woeginger2008axiomatic}. 
On the other hand, a ``tempered'' version $\bar{\nu}$ of the $\nu$-index, modified so as to be not larger than the number of published papers, satisfies the inequalities
\begin{equation*}
h\le \bar{\nu}\le g.
\end{equation*}
We will finish off by introducing a more general family of citation indexes $(\nu_\alpha)$, where $\nu_\alpha$ is a non-decreasing (integer-valued) function of a real parameter $\alpha\ge 0$. Here, $h=\nu_0$ and $\nu=\nu_1$, while  the limiting value $\nu_\infty$ can be expressed in terms of the maximum citation. 

\subsection{\bf Disclaimer}

There have been a lot of discussions about the utility and limitations of the indexes $h$ and $g$ (see, e.g., \citealp{alonso2009h,{Costas2007},hirsch2007index, web-h,Thelwall2025, waltman2016}, and references therein), including their questionable predictive power. In this work, we only interpret citation indexes as a suitable characteristic of productivity. However, some further thoughts about the societal dimension of citation indexes will be added in the Conclusion section.

\section{The $h$ and $g$ indexes}

\subsection{\bf Notation}

Let us fix some  notation. Suppose that an author has published $m\ge1$ papers, with the ordered numbers of citations $x_1\ge 
\dots \ge x_m$, where $x_i\ge0$ are integers (possibly zero). We call  $\boldsymbol{x}=(x_1,
\dots,x_m)$ the \emph{citation vector}. 
A zero vector $\boldsymbol{0}=(0,\dots,0)$ represents the degenerate case of no citations (note that the variable dimension of this vector is determined by the number of published papers). 
Denote
\begin{equation}\label{eq:sk}
S_k=x_1+\dots+x_k,\qquad k=1,\dots,m.
\end{equation}
Clearly, $n=S_m$ is the total number of citations generated by the $m$ papers.
Furthermore, we write  
\begin{equation}\label{eq:m*}
m_{*\myn}(j)=\sum_{i=1}^m \boldsymbol{1}_{\{x_i\ge j\}}=\#\{x_i\ge j\}
\end{equation}
for the number of papers with at least $j$ citations each; here, 
$\boldsymbol{1}_{\myn A}=1$ if condition $A$ is satisfied and $\boldsymbol{1}_{\myn A}=0$ otherwise.

Following \cite{woeginger2008axiomatic}, we say that a vector $\boldsymbol{x}=(x_1,\dots,x_m)$ is \emph{dominated} by a vector $\boldsymbol{y}=(y_1,\dots,y_{\ell})$ (written as $\boldsymbol{x}\preccurlyeq\boldsymbol{y}$) if $x_i\le y_i$ for all $i\ge1$; more precisely, if $m\le \ell$ then $x_i\le y_i$ for $i\le m$, but if $m>\ell$ then $x_i\le y_i$ for $i\le \ell$ and $x_i=0$ for $\ell <i\le m$. Effectively, these two cases imply that we complement the absent components of either $\boldsymbol{x}$ or $\boldsymbol{y}$ with fictitious zeros to equalize their dimensions, and then the dominance holds component-wise.

\begin{remark}
The component-wise dominance $\boldsymbol{x}\preccurlyeq\boldsymbol{y}$ should not be confused with \emph{(weak) majorization} $\boldsymbol{x}\prec_w\boldsymbol{y}$, defined by the conditions $\sum_{i=1}^kx_i\le \sum_{i=1}^ky_i$, for all $k$ (see \citealp[pp.\ 11--12]{Marshall}). Like before, the lengths of vectors $\boldsymbol{x}$ and $\boldsymbol{y}$ are equalized by adding fictitious zeros as necessary. Clearly, if $\boldsymbol{x}\preccurlyeq\boldsymbol{y}$ then $\boldsymbol{x}\prec_w\boldsymbol{y}$, but not conversely. 
\end{remark}

\subsection{\bf Generic properties}

The following natural conditions are commonly assumed 
for any reasonable citation index $c(\boldsymbol{x})$ 
\citep{woeginger2008axiomatic} (cf.\ \citealp{woeginger2008axH}):
\begin{itemize}
\item[(C1)] If $\boldsymbol{x}=\boldsymbol{0}$ then $c(\boldsymbol{x})=0$.

\item[(C2)] If $\boldsymbol{x}=(x_1, 
\dots, x_m)$ and $\boldsymbol{y}=(x_1, 
\dots, x_m,0)$ then $c(\boldsymbol{x})=c(\boldsymbol{y})$.

\item[(C3)] 
If $\boldsymbol{x}\preccurlyeq \boldsymbol{y}$ then $c(\boldsymbol{x})\le c(\boldsymbol{y})$.
\end{itemize}

\begin{remark}
The majorization relation $\prec_w$ looks more flexible as a comparative tool. The corresponding version of property \textup{(C3)} is stated similarly:
\begin{itemize}
\item[\textup{(C3$'$)}]
\textup{If $\boldsymbol{x}\prec_w\boldsymbol{y}$ then $c(\boldsymbol{x})\le c(\boldsymbol{y})$.}
\end{itemize}
However (perhaps, surprisingly), the $h$-index does not satisfy \textup{(C3$'$)}: e.g., for $\boldsymbol{x}=(2,2)$ and $\boldsymbol{y}=(8,1)$ we have $\boldsymbol{x}\prec_w\boldsymbol{y}$ but $h(\boldsymbol{x})=2>h(\boldsymbol{y})=1$. On the other hand, the $g^*\!$-index (but not $g$) does satisfy \textup{(C3$'$)} (see definitions \eqref{eq:gindexDF} and \eqref{eq:gindex*} below); e.g., $g^*\mynn(\boldsymbol{x})=2<g^*\mynn(\boldsymbol{y})=3$. 
\end{remark}

\subsection{\bf Mathematical expressions and relations for $h$ and $g$}
Let us now recall the above verbal definitions of the $h$ and $g$ indexes and put them into an explicit mathematical formulation. 
Starting with the $h$-index, its definition can be expressed as follows,
\begin{equation}
    h\equiv h(\boldsymbol{x})=\max \left\{ j\ge 1 \colon\! \sum_{i=1}^{m} \boldsymbol{1}_{\{x_i\ge j\}}
    \ge j \right\}\!,
\label{eq:hindexDF}
\end{equation}
or, using notation \eqref{eq:m*}, \begin{equation}
 h=\max \bigl\{ j\ge 1 \colon m_{*\myn}(j)\ge j\bigr\}.
\label{eq:hindexDF*}
\end{equation}

Note that the maximum in \eqref{eq:hindexDF} is uniquely defined, since the sum on the left-hand side is a decreasing function of $j$, while the right-hand side of the testing inequality is strictly increasing.

In particular, noting that 
$m_{*\myn}(j)\le m$, it follows that the $h$-index is bounded by the number of papers:
\begin{equation*}
h\le m.
\end{equation*}
In the degenerate case with $x_i\equiv 0$ (i.e., $\boldsymbol{x}=\boldsymbol{0}$), the inequality in \eqref{eq:hindexDF} is only satisfied for the value $j=0$, which is excluded from the testing range; thus, the resulting set of suitable $j$'s is empty and, according to the common convention, its maximum is set to be zero: $\max\varnothing = 0$; hence
$h(\boldsymbol{0})=0$,
so that property (C1) is automatically satisfied. It is also easy to see that (C2) holds (because $h$ is insensitive to zero citations) and that (C3) is also true.


%
%
Next, 
the definition of the $g$-index can be written 
as follows,
\begin{equation}
\label{eq:gindexDF0}
g\equiv g(\boldsymbol{x})=\max\left\{1\le k\le m\colon \sum_{i=1}^{k}x_{i}\ge k^2\right\}\!,
\end{equation}
or, recalling notation \eqref{eq:sk},
\begin{equation}
\label{eq:gindexDF}
 g=\max\left\{1\le k\le m\colon S_k\ge k^2\right\}\!.
\end{equation}
Note that if $\boldsymbol{x}=\boldsymbol{0}$ (i.e., all $x_i=0$) then the set under the max-symbol is empty, in which case, by the same convention, we define the maximum as zero. That is to say, the $g$-index for the zero citation vector equals zero:
\begin{equation}\label{eq:g0}
g(\boldsymbol{0})=0.
\end{equation}
Also note that, because the testing range of $k$'s in  \eqref{eq:gindexDF} is bounded by $m$, we must have 
\begin{equation*}
g\le m.
\end{equation*}

It can be shown that the $g$-index is not smaller than the $h$-index of the same author \citep[Proposition I.2. p.\,133]{egghe2006theory}, 
\begin{equation*}
h\le g.   
\end{equation*} 
Indeed, if the $h$-index has value $h$ then there are $h$ papers with at least $h$ citations each, and therefore with at least $h\times h=h^2$ citations in total. Hence, the trial value $k=h$ satisfies the inequality condition in \eqref{eq:gindexDF}, which implies that $g\ge k=h$, as claimed. 


\subsection{\bf Auxiliary lemmas for sums}

According to definition \eqref{eq:gindexDF}, the index $g$ is the largest value of $k\le m$ for which $S_k\ge k^2$. But it may be unclear whether the inequality $S_k\ge k^2$ can fail for some $k<g$. Let us show that $S_k\ge k^2$ \emph{for all} $k\le g$. 

\begin{lemma}\label{lm:1}
If $S_k<k^2\mynn$ for some $k\ge 1$, then $S_{\ell}<\ell^2\mynn$ for all $\ell\ge k$. 
\end{lemma}
\begin{proof}
Using that $x_k=\min\{x_1,\dots,x_k\}$, we have 
$$
k^2\mynn>S_k=x_1+\dots+x_k\ge k\myp x_k,
$$
which implies that $k>x_k\ge x_{k+1}$. Hence,  
$$
S_{k+1}=S_k+x_{k+1}<k^2+k<(k+1)^2,
$$
that is, $S_{k+1}<(k+1)^2\myn$. The general claim then follows by induction. 
\end{proof}

\begin{lemma}\label{lm:2}
If $S_k<k^2\mynn$ then $g<k$. In particular, 
$g<m$ or $g=m$ according as $S_m<m^2\mynn$ or $S_m\ge m^2\mynn$, respectively.
\end{lemma}
\begin{proof} Readily follows by Lemma \ref{lm:1} and definition \eqref{eq:gindexDF}.
\end{proof}

\subsection{\bf 
The unconstrained index $g^*$}
Turning to the verification of the required properties (C1)--(C3) for the $g$-index, we see that (C1) automatically holds due to \eqref{eq:g0}. It is also easy to see that (C3) holds as well. However, the result of Lemma \ref{lm:2} suggests, a bit surprisingly, that property (C2) may fail. For instance, for $\boldsymbol{x}=(4)$ we have $g(\boldsymbol{x})=1$, but for $\boldsymbol{y}=(4,0)$
definition \eqref{eq:gindexDF} yields $g(\boldsymbol{y})=2$.

To salvage (C2), and also to amplify the role of top-cited papers, it was suggested \citep{egghe2006theory, woeginger2008axiomatic} to lift the constraint $k\le m$ in the definition of the $g$-index (see \eqref{eq:gindexDF}) by complementing the citation vector $\boldsymbol{x}=(x_1,\dots,x_m)$ with  additional zeros, as if such fictitious papers have been published but generated no citations: $\boldsymbol{x}'=(x_i')=(x_1,\dots,x_m,0,\dots)$. We denote this version of $g$ by $g^*\mynn$:
\begin{equation}
\label{eq:gindex*}
g^*\mynn\equiv g^*\mynn(\boldsymbol{x})=\max\left\{k\ge 1\colon \sum_{i=1}^{k}x'_{i}\ge k^2\right\}
\end{equation}
or, equivalently,
\begin{equation}
\label{eq:gindex**}
g^*\mynn=\max\bigl\{k\ge 1\colon S_k\ge k^2\bigr\},
\end{equation}
where we define $S_k=S_m$ for all $k\ge m$.

Comparing definitions \eqref{eq:gindexDF} and \eqref{eq:gindex*}, we see that 
\begin{equation*}
g\le g^*\mynn,
\end{equation*}
and moreover, if $g<m$ then $g=g^*\mynn$. However, the case where $g=m$ may be drastically different. 
\begin{lemma}
Suppose that $S_m\ge (m+1)^2\myn$. Then $g=m$ but $g^*=\displaystyle \bigl\lfloor\sqrt{S_m}\bigr\rfloor\ge m+1$.
\end{lemma}
\begin{proof} 
Note that $g=m$ by Lemma \ref{lm:2}. 
By definition \eqref{eq:gindex**} we have 
$$
S_{g^*}\mynn=S_m \ge (g^*)^2,\qquad 
S_{g^*\myn+1}=S_m <(g^*\!+1)^2.
$$
In turn, this  implies  
$$
\sqrt{S_m}-1<g^*\!\le \sqrt{S_m},
$$ 
that is,  $g^*\!=\displaystyle \bigl\lfloor\sqrt{S_m}\bigr\rfloor\ge m+1$, as claimed.
\end{proof}

For example, for $\boldsymbol{x}=(5,4)$ we have $m=2$, $S_m=9$, $g=2$ and $g^*\!=3$. A striking real-life example illustrating this situation is the case of John Nash (see \citealp{woeginger2008axiomatic}), with the (rounded) citation vector
$\boldsymbol{x}=(2000, 2000, 1500, 1000, 400, 250, 100, 100)$, for which we get $g=8$ but $g^*\!=85$.

\section{An alternative citation index $\nu$}
\subsection{\bf Idea and definitions}

Trying to reconcile the definition of the $h$-index given by formula  \eqref{eq:hindexDF}, with the definition of the $g$-index in \eqref{eq:gindexDF0} by taking into account the actual citations of the top papers, we propose a new citation index called  the \emph{$\nu$-index}, defined as the maximum integer $\nu$ such that the total sum of citation counts of papers with at least $\nu$ citations each is not less than $\nu^2\myn$. Mathematically, this is expressed as (cf.\ \eqref{eq:hindexDF})
\begin{equation}
  \nu\equiv \nu(\boldsymbol{x})=\max\left\{j\ge 1\colon \sum_{i=1}^{m} x_i\myp \boldsymbol{1}_{\{x_i\ge j\}}\ge j^2 \right\}\!,
\label{eq:gindexDF-old}
\end{equation}
or, equivalently,
\begin{equation}
  \nu=\max\bigl\{j\ge 1\colon
S_{m_{*\myn}(j)}\ge j^2 \bigr\}.
\label{eq:gindexDF-old*}
\end{equation}%

Similarly to \eqref{eq:hindexDF}, the maximum is uniquely defined, noting that the sum in \eqref{eq:gindexDF-old}
is a decreasing function of $j$, while the right-hand side is strictly increasing. It is also worth pointing out that, unlike $g$ vs.\ $g^*\mynn$, the $\nu$-index is insensitive to fictitious zeros. 

Simple examples show that the value of $\nu$ may be larger than the total number of papers, in contrast with the $h$ and $g$ indexes. For instance, for $\boldsymbol{x}=(9,7
,1)$ we get $\nu(\boldsymbol{x})=4>3$. 

Clearly, this occurs because our definition of $\nu$ gives prominence to few highly cited papers. If unwanted, this can be suppressed by modifying the definition via an explicit constraint $\nu\le m$:
\begin{equation}\label{eq:nu-bar}
 \bar{\nu}\equiv \bar{\nu}(\boldsymbol{x})=\max\left\{1\le j\le m\colon \sum_{i=1}^{m} x_i\myp \boldsymbol{1}_{\{x_i\ge j\}}\ge j^2 \right\}\!,
 \end{equation}
or, equivalently,
\begin{equation*}
 \bar{\nu}=\max\bigl\{1\le j\le m\colon S_{m_{*\myn}(j)}\ge j^2 \bigr\}.
\end{equation*}
We call $\bar{\nu}$ a \emph{tempered $\nu$-index}.

\subsection{\bf Checking the basic properties}
\begin{lemma}
    The indexes $\nu$ and $\bar{\nu}$ satisfy the basic properties \textup{(C1)--(C3)}.
\end{lemma}
\begin{proof}
Properties (C1) and (C2) are straightforward, since $\nu$ and $\bar{\nu}$ are insensitive to zero values $x_i=0$. Monotonicity (C3) is also obvious because the sums in \eqref{eq:gindexDF-old} and \eqref{eq:nu-bar} are monotone increasing in each component $x_i$. 
\end{proof}

The $\nu$-index combines the features of both the $h$-index and the $g$-index. It takes into account citations that are equal to or greater than a minimum threshold value of $\nu$ as in the $h$-index, while also including higher citations as in the $g$-index. This ensures that the $\nu$-index captures the impact of highly cited papers and provides a more balanced picture of their overall scholarly impact. In particular, it may be expected that the $\nu$-index interpolates between $h$ and $g$. The next result supports this conjecture.

\subsection{\bf Main result -- ordering relations between the indexes}

\begin{theorem}\label{th:1}
The citation indexes $h$, $\nu$, $\bar{\nu}$, $g$ and $g^*\!$ are in the following ordering relations:
\begin{equation}
\label{eq:hgnu}
h \le \nu \le g^*\mynn,
\qquad h \le \bar{\nu} \le g.
\end{equation}
\end{theorem}
\begin{proof} We only prove the inequalities for $\nu$;
the proof 
for $\bar{\nu}$ is similar. 
First, by definition of $h$ in \eqref{eq:hindexDF} and \eqref{eq:hindexDF*} we can write
$$
h\le m_{*\myn}(h)=\sum_{i=1}^m \boldsymbol{1}_{\{x_i\ge h\}}
\le\frac{1}{h} \sum_{i=1}^m x_i\myp \boldsymbol{1}_{\{x_i\ge h\}}=\frac{1}{h}S_{m_{*\myn}(h)}.
$$
Hence,
$$
S_{m_{*\myn}(h)}=\sum_{i=1}^m x_i\myp \boldsymbol{1}_{\{x_i\ge h\}}\ge h^2,
$$
which implies, according to \eqref{eq:gindexDF-old*}, that $\nu\ge h$, as claimed.


Next, 
the maximizing sum in \eqref{eq:gindexDF-old} 
is expressed as 
\begin{equation}\label{eq:J}
\nu^2\mynn \le \sum_{i=1}^{m} x_i\myp \boldsymbol{1}_{\{x_i\ge \nu\}}=S_{m_{*\myn}(\nu)},
\end{equation}
thus involving $m_{*\myn}(\nu)$ terms $x_{i}$
satisfying the inequality \begin{equation}\label{eq:>ell}
x_{i}\ge \nu, \qquad i=1,\dots,m_{*\myn}(\nu).
\end{equation}
If $m_{*\myn}(\nu)\le \nu$ then from \eqref{eq:J} we obtain (adding fictitious zeros if $\nu>m$)
$$
\nu^2\mynn\le S_{m_{*\myn}(\nu)}\le S_{\nu},
$$
and it follows from  definition \eqref{eq:gindex**} that $g^*\!\ge \nu$. Alternatively, if $m_{*\myn}(\nu)\ge \nu$ then, using \eqref{eq:>ell}, we can write
$$
S_{m_{*\myn}(\nu)}\ge S_\nu\ge \nu^2\mynn,
$$
and, as before, it follows  that $g^*\!\ge \nu$.
\end{proof}

\subsection{\bf R code and some simple examples}

A simple R code to calculate various indexes is given below:

\medskip

\begin{minipage}{\hsize}%
\lstset{language=R} \begin{lstlisting}[frame=single,framexleftmargin=-1pt,framexrightmargin=-17pt,framesep=12pt,linewidth=0.98\textwidth,showstringspaces=false]
# indexes h, nu, nu.bar, g, g.star
ind <- function(x) # x = input citation vector
{ x <- sort(x, decreasing = TRUE)  # ordering 
  m <- length(x)  # number of papers
# h
  h <- 0
  while (h < length(x) && x[h + 1] >= h + 1) 
  { h <- h + 1
  }
\end{lstlisting}
\end{minipage}

\begin{minipage}{\hsize}%
\lstset{language=R} \begin{lstlisting}[frame=single,framexleftmargin=-1pt,framexrightmargin=-17pt,framesep=12pt,linewidth=0.98\textwidth,showstringspaces=false]
# nu
  nu <- 0
  while (sum(x[which(x >= (nu + 1))]) 
         >= (nu + 1)^2) 
  { nu <- nu + 1
  }
# nu.bar
  nu.bar <- min(nu,m) 
# g
  g <- max(which(cumsum(x) >= (1:m)^2))
  # g.star 
  if (sum(x) >= m^2)
  { g.star <- floor(sqrt(sum(x)))
  } 
  else 
  { g.star <- max ( which (cumsum(x) >= (1:m)^2))
  }
}# Printing the output:
cat("x =","(",x,");", "\n")
cat("h =", h, "nu.bar =", nu.bar, "nu =", nu, 
     "g =", g, "g.star =", g.star)
\end{lstlisting}
\end{minipage}

\bigskip
\noindent
{\em Example} (John Nash case):

\medskip
\begin{minipage}{\hsize}
\lstset{language=R} \begin{lstlisting}[frame=single,framexleftmargin=-1pt,framexrightmargin=-17pt,framesep=12pt,linewidth=0.98\textwidth]
x <- c(2000,2000,1500,1000,400,250,100,100)
ind(x)
# x = ( 2000 2000 1500 1000 400 250 100 100 ); 
# h = 8 nu.bar = 8 nu = 85 g = 8 g.star = 85
\end{lstlisting}
\end{minipage}

\medskip\bigskip\noindent
The following Table \ref{tab:example} presents the various citation indexes for a few simple examples.
\begin{table}[h!]
\caption{Illustrative examples of different indexes.}       \label{tab:example}
        \centering
        \begin{tabular}{|l||c|c|c|c|c|}
            \hline
    $\boldsymbol{x}=(x_1,\dots,x_m)$ & $h$& $\bar{\nu}$ & $\nu$ & $g$ &$\textstyle g^*$ \\
            \hline
 $(3,2,2,2)$ & 2&2&2&2&2\\
            $(12,3,1)$ &2 & 3&3 &3 &4\\
           $(12,3,1,0)$ &2 & 3&3 &4 &4\\
          $(6,3,1,0)$ & 2 & 3 & 3&3&3
 \\          $(5,3,2,1)$ & 2 & 2 & 2&3&3\\       
   $(8,1,1)$ &1 &2 &2 & 3&3\\
          $(8,4,3,2,1)$ &3 &3 &3 & 4&4 \\
          $(18,18,1,1)$ & 2& 4&6 &4 &6\\
          $(20, 20, 18, 6,1,0)$& 4 & 6& 7& 6& 8 \\
          \hline
        \end{tabular}
     \end{table}

\subsection{\bf Cases of equality}
One observation from Table \ref{tab:example} is that, occasionally, some of the indexes may coincide, which warrants a question of exploring the cases of equalities in \eqref{eq:hgnu}. The possible equality $h=g^*\!$ was addressed by \citet{Egghe2019}.

\begin{theorem}
\label{th:2}
The equalities in the index inequalities \eqref{eq:hgnu} of Theorem \ref{th:1} hold if and only if the following conditions are satisfied, respectively\/\textup{:}

\medskip
\begin{tabular}{rll}
{\rm (a)}&\hspace{-.3pc} $h=\nu\colon$& $S_{m_{*\myn}(h+1)}<(h+1)^2;$ \\[.4pc]
{\rm (b)}&\hspace{-.3pc} $\nu=g^*\!{\colon}$& $S_{\nu+1}<(\nu+1)^2;$ \\[.4pc]
{\rm (c)}&\hspace{-.3pc} $h=\bar{\nu}\colon$& $h=m$\textup{,} or $h<m$ and $S_{m_{*\myn(h+1)}}<(h+1)^2;$ \\[.4pc]
{\rm (d)}&\hspace{-.3pc} $\nu=g\colon$& $\nu=m$\textup{,} or $\nu<m$ and $S_{\nu+1}<(\nu+1)^2.$
\end{tabular}
%
%
%

\end{theorem} 
\begin{proof}
Since it is always true that $\nu\ge h$ (see \eqref{eq:hgnu}), the equality $\nu=h$ simply means that $\nu<h+1$. But, according to definition \eqref{eq:gindexDF-old*}, the latter inequality is equivalent to 
$S_{m_{*\myn}(h+1)}<(h+1)^2$, which is the claim of part (a). Essentially the same argument proves part (c), except that, due to the bound $\bar{\nu}\le m$, a special case arises if $h=m$, which automatically implies $\bar{\nu}=m$. 

Similarly, due to \eqref{eq:hgnu} we have $g^*\!\ge \nu$, while the inequality $g^*\!<\nu+1$ is equivalent to $S_{\nu+1}<(\nu+1)^2$, according to definition \eqref{eq:gindex**}, and the claim of part (b) follows. The same argument applies to part (d), with an additional consideration of the special case $\nu=m$. 
\end{proof}

Part (a) is  exemplified by  $\boldsymbol{x}=(3,2,1)$: here, $h=\nu=2$, while $m_{*\myn}(2)=2$, 
$m_{*\myn}(3)=1$ and $S_2=5>2^2$ but $S_1=3<3^2$. The same example gives $g^*\!=2$, confirmed by the inequality $S_3=6<3^2$, in line with part (b).  Furthermore, since $g=g^*\!$ and $\bar{\nu}=\nu$, this example also illustrates parts (c) and (d). As for the boundary case of parts (c) and (d), it occurs, for example, for $\boldsymbol{x}=(4,3,3)$, where $h=\bar{\nu}=\nu=g=3$. Another example of (b) is the John Nash case mentioned above.

\subsection{\bf Data example}

Here, we illustrate the calculation of the various citation indexes for real data collected by the first-named author (available online at  \href{https://github.com/Ruheyan/WoS-citation-data/tree/main}{https:\allowbreak //\allowbreak github.\allowbreak com/\allowbreak Ruheyan/\allowbreak WoS-\allowbreak citation-\allowbreak data/\allowbreak tree/\allowbreak main}).\footnote{Similar  citation data 
was used in a conference paper \citep{nuerISSI2021} and PhD thesis \citep{Nuer2023}.} The dataset comprises citation counts, with a cut-off date of 19th September 2022,
of 3,615 papers (with 73,730 citations in total) of 111 authors who published 
a paper in the first 10 issues of 
\emph{Electronic Journal of
Probability} (EJP), 
vol.\,24 (2019) 
(\href{https://projecteuclid.org/journals/electronic-journal-of-probability/volume-24/issue-none}{https:\allowbreak //\allowbreak projecteuclid.\allowbreak org/\allowbreak journals/\allowbreak electronic-\allowbreak journal-\allowbreak of-\allowbreak probability/\allowbreak volume-\allowbreak 24/\allowbreak issue-\allowbreak none}). The data 
were derived from the 
Web of Science \citep{WoS}. 


\begin{figure}
\centering
\includegraphics[width=1\textwidth]{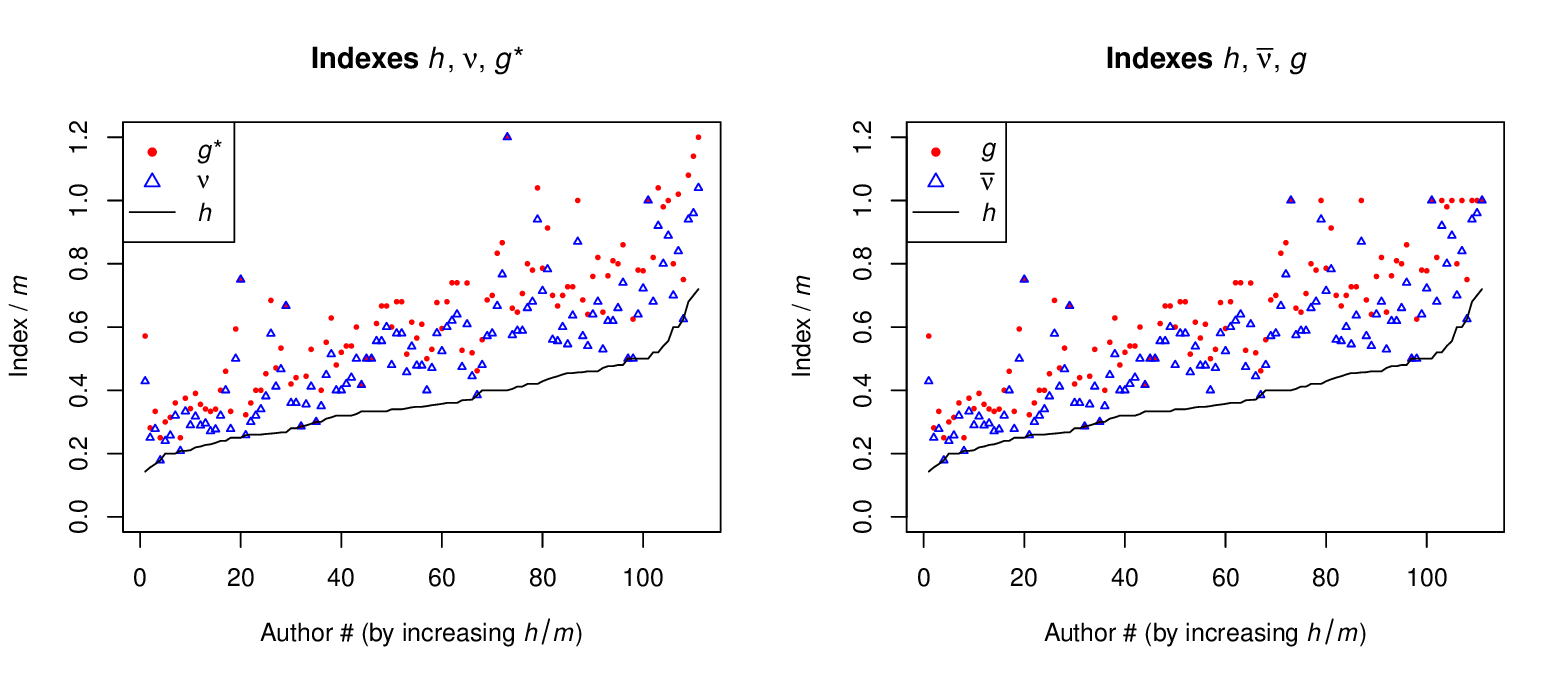}
\caption{Index plots for the EJP dataset, showing triplets of indexes ($h\le \nu\le g^*$ or $h\le\bar{\nu}\le g$) normalized by the number of published papers per author. The authors are ranked in increasing order with respect to the parameter $h/m$.}
\label{fig1}
\end{figure}

Fig.\ \ref{fig1} shows the plots representing the indexes $h$, $\bar{\nu}$, $\nu$, $g$, and $g^*$ (in triplets, for ease of comparison) for all 111 authors, normalized  
by the number of papers per author. The calculated values confirm the inequalities of Theorem \ref{th:1}, but one can observe that the new index ($\nu$ or $\bar{\nu}$) tends to be closer to the upper bound $g^*\!$ or $g$, respectively. Furthermore, Table \ref{tab2} presents correlations between different indexes\,---\,not surprisingly, they are all strongly positively correlated (especially in the ``sister'' pairs $(\nu,\bar{\nu})$ and $(g,g^*)$), but correlation with the number of papers ($m$) is weaker.

\medskip
\begin{table}[h!]
\caption{Pairwise correlations across  the citation indexes and the number of published papers ($m$).}       \label{tab2}
        \centering
        \begin{tabular}{|c|cccccc|}
           \hline
   {\bf index} &$h$& $\nu$ & $\bar{\nu}$ & $g$ &$g^*$&$m$ \\
\hline
$h$ 
&1.0000
&0.9649  
&0.9646  
&0.9656  
&0.9725 
&0.8044\\
$\nu$ 
&0.9649  
&1.0000
&0.9998 
&0.9932 
&0.9978 
&0.7743\\
$\bar{\nu}$ 
&0.9646  
&0.9998  
&1.0000 
&0.9942 
&0.9978  
&0.7768\\
$g$ 
&0.9656  
&0.9932 
&0.9942  
&1.0000 
&0.9967 
&0.8046\\
\ $g^*$ 
&0.9725  
&0.9978  
&0.9978 
&0.9967  
&1.0000
&0.7893\\
$m$ 
&0.8044 
&0.7743  
&0.7768 
&0.8046  
&0.7893 
&1.0000\\
\hline
        \end{tabular}
     \end{table}



\section{Parametric family $(\nu_\alpha)$}

\subsection{\bf Definition and monotonicity}

It is quite natural to generalize the definition of the index $\nu$ in  \eqref{eq:gindexDF-old} by considering different powers. Namely, for $\alpha\ge0$ we define the $\nu_\alpha$-index as
\begin{equation}\label{eq:nu_alpha}
\nu_\alpha\equiv \nu_\alpha(\boldsymbol{x})=\max\left\{j\ge 1\colon \sum_{i=1}^{m} x_i^\alpha\myp \boldsymbol{1}_{\{x_i\ge j\}}\ge j^{\alpha+1}\right\}\!.
\end{equation}
Clearly, for $\alpha=0$ and $\alpha=1$ this definition is reduced to \eqref{eq:hindexDF} and \eqref{eq:gindexDF-old}, respectively:
\begin{equation*}
\nu_0=h, \qquad \nu_1=\nu.
\end{equation*}
Like in \eqref{eq:hindexDF} and \eqref{eq:gindexDF-old}, the existence and uniqueness of the maximum in \eqref{eq:nu_alpha} is self-evident, noting that the sum is a decreasing function of $j$ while the right-hand side is strictly increasing. It is straightforward to verify that $\nu_\alpha$ satisfies (C1)--(C3). We also observe the monotonicity of the family $(\nu_\alpha)$.
\begin{theorem}\label{th:3}
The function $\nu_\alpha$ is increasing in $\alpha\ge0$.
\end{theorem}
\begin{proof}
Rewrite \eqref{eq:nu_alpha} as 
\begin{equation}\label{eq:nu-alpha-norm1}
\nu_\alpha=\max\left\{j\ge 1\colon \sum_{i=1}^{m} \left(\frac{x_i}{j}\right)^{\mynn \alpha} \!\boldsymbol{1}_{\{x_i\ge j\}}\ge j\right\}\!,
\end{equation}
and note that the sum in \eqref{eq:nu-alpha-norm1} is monotone increasing in $\alpha$, since $x_i/j\ge 1$. 
\end{proof}

As an illustration of sensitivity and fluidity of $\nu_\alpha$, in the John Nash case it is easy to check that, for example, for $\alpha=0.5$ we have $\nu_{0.5}=35$, 
compared to $\nu_1\mynn=g^*\!=85$. 
R code to calculate $\nu_\alpha$ is given below:

\medskip
\begin{minipage}{\hsize}%
\lstset{language=R}
\begin{lstlisting}[frame=single,framexleftmargin=-1pt,framexrightmargin=-17pt,framesep=12pt,linewidth=0.98\textwidth,showstringspaces=false]
# nu.alpha
x <- sort(x, decreasing=TRUE)
nu.alpha <- function(alpha)
{ sapply(alpha, function(a) 
  { nu <- 0
    while (sum((x[x >= (nu + 1)] / (nu + 1))^a)
          >= (nu + 1)) 
    { nu <- nu + 1
    }
    return(nu)
  })
}
# Plotting the output:
curve(nu.alpha, col = "red", lwd = 2, 
      xlim = c(0, max(x)+20), ylim = c(1, max(x)), 
      xlab = expression(paste(alpha)),
      ylab = expression(paste(nu[alpha])),
      main = bquote(paste(bold("x "), "= (", 
      .(toString(x)), ")")))
\end{lstlisting}
\end{minipage}

\medskip\bigskip\noindent
The next Fig.\ \ref{fig2} illustrates the behavior of the function $\nu_\alpha$ for some examples from Table \ref{tab:example}. The reader may also find it interesting to run this code on the citation data of John Nash.


\begin{figure}[h!]
\centering
\includegraphics[width=1\textwidth]{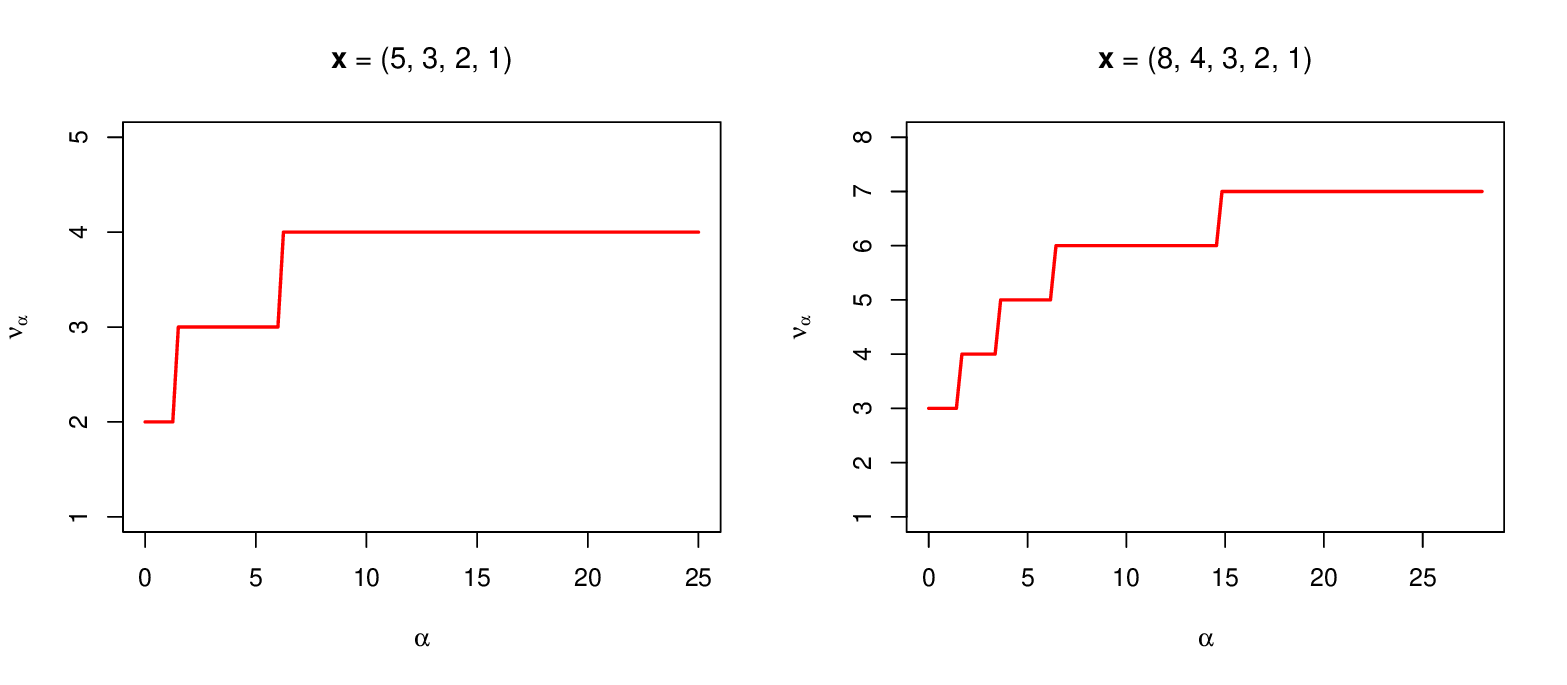}
\caption{{Illustrative graphs of the index $\nu_\alpha$ as a function of parameter $\alpha\in[0,\infty)$. Note the values $\nu_0=h$, $\nu_1=\nu$, and $\nu_\infty=x_1-1$ (cf.\ \eqref{eq:nu-inf}).}}
\label{fig2}
\end{figure}
\subsection{\bf The limit as $\alpha\to\infty$}
It is 
interesting to understand the meaning of the limiting value $\nu_\infty=\lim_{\alpha\to\infty} \nu_\alpha$.

\begin{theorem}\label{th:4}
For a citation vector $\boldsymbol{x}=(x_1,\dots,x_m)$, denote by $\ell_1=\sum_{i=1}^m \boldsymbol{1}_{\{x_i=x_1\}}\equiv m_{*\myn}(x_1)$ the multiplicity of the top citation $x_1=\max\{x_i,\,1\le i\le m\}$. Then 
\begin{equation}
\label{eq:nu-inf}
\nu_\infty(\boldsymbol{x})=
\begin{cases}
x_1-1&\text{if }\,\ell_1<x_1,\\
x_1&\text{if } \,\ell_1\ge x_1.
\end{cases} 
\end{equation}
\end{theorem}
\begin{proof}
Follows using \eqref{eq:nu-alpha-norm1} by noting that $(x_1/j)^\infty$ equals $\infty$, $1$ or $0$ according as $j<x_1$, $j=x_1$ or $j>x_1$, respectively.
\end{proof}

\section{Conclusion}

We have introduced some new citation indexes starting with $\nu=\nu_1$, and investigated their relations with the classical indexes $h$ and $g$. As already mentioned, the $h$-index is straightforward and informative, but it is limited by only acknowledging the fact of a high citation but not the actual number of citations. In contrast, the $g$-index is based exclusively on the citations of a few top papers, but ignoring the ``footing'' of lower-cited papers. 

Our synthetic proposal of the $\nu$-index is designed so as to take into account both higher and lower cited papers, which may assess the individual's productivity in a more fair and balanced way. Indeed, we have seen that the $\nu$-index is in a sense bridging Hirsch's $h$ and Egghe's $g$. Furthermore, the spectrum of the indexes $(\nu_\alpha)$ provides a flexible toolkit that allows one either to enhance or to inhibit the input from top-cited papers, as required. 

Of course, it goes without saying that none of these, or any other indexes known in the literature, is perfect and should replace the rest. 
In fact, a reasonable practical recommendation may be to choose a few indexes to judge someone's academic achievement, depending on the assessment requirements and also on the specific features of the scientific domain. 
In this regard, it may be useful to choose the parameter $\alpha$ in the index $\nu_\alpha$ according to certain individual features of the citation vector $\boldsymbol{x}$, in the spirit of limit theorems for norms of random vectors \citep{Bogachev2006,Schlather}. We will address this issue in our future work. 

In conclusion, we 
reiterate that prudence, maturity and  care should be exercised when using citation indexes in social practice, especially making sure to avoid misuse and/or abuse of their utility as predictors of future performance and productivity. Although citation indexes succinctly grasp some objective aggregated information from citation records, they are deceptively easy to compute,  replacing individual research track records with a simple number, while these results should be verified and complemented by human evaluation by experts. 

The scientometrics community has quickly realized, and extensively documented, the growing threat of misusing the $h$-index and other indicators for far reaching and
often unjustified implications in the social interpretation  (see, e.g., \citealp{alonso2009h,Costas2007,Hicks2015, Thelwall2025, {waltman2016}}, and further references therein). These concerns and wide discussions have led to the creation and promotion of good practice protocols, such as the San Francisco Declaration on Research Assessment 
\citep{DORA} or the Leiden Manifesto \citep{Hicks2015}. 

The risks are further amplified by the fast growing use of Artificial Intelligence (AI) including Large Language Models such as ChatGPT, whereby the responsibility for conclusions and extrapolations may be
delegated inadvertently to the computer \citep{Thelwall2025}. Although deployment of AI for assistance in technical analyses and summarization is an inevitable and welcome trend, the best vaccine against misuse and abuse is to combine formal calculations and summaries with a robust comparison against the specific domain ``golden standards'', based on an objective expert evaluation and enhanced by a reproducible and unbiased statistical analysis.


\subsection*{Funding}
R.N.\ acknowledges funding from the School of Mathematics at the University of Leeds for attending the STI-ENID\,2025 conference to present a poster based on this work.

\subsection*{Author contributions}
L.B.\ and R.N.\ designed the study and wrote the first draft of the paper. L.B.\ and J.V.\ contributed to the mathematical derivation and interpretation of the results. 
All authors agreed with the final version of the paper.

\subsection*{Previous presentation}
These results were 
presented as a poster 
at the 29th Annual International Conference on Science 
and Technology Indicators (STI-ENID\,2025), Bristol, UK,
3--5 September 2025 (\href{https://virtual.oxfordabstracts.com/event/74626/submission/143}{https:\allowbreak //\allowbreak virtual.\allowbreak oxfordabstracts.\allowbreak com/\allowbreak event/\allowbreak 74626/\allowbreak submission/\allowbreak 143}).



\subsection*{Data availability}
EJP dataset used in this article is available online at \href{https://github.com/Ruheyan/WoS-citation-data/tree/main}{https:\allowbreak //\allowbreak github.\allowbreak com/\allowbreak Ruheyan/\allowbreak WoS- \allowbreak citation-\allowbreak data/\allowbreak tree/\allowbreak main}.

\bibliographystyle{plain}


\begin{thebibliography}{99}

\makeatletter
\renewcommand{\@listI}%
{\leftmargin=\parindent
\partopsep=0pt
\parskip=-5pt
\topsep=3pt
\itemsep=3pt
\labelwidth=\leftmargini }
\renewcommand{\@listii}{\setlength{\topsep}{1pt}
} \setlength{\itemsep}{1pt} \setlength{\partopsep}{-5pt}
\setlength{\parskip}{2pt} \makeatother

\expandafter\ifx\csname natexlab\endcsname\relax\def\natexlab#1{#1}\fi


\bibitem[Alonso et al.(2009)]{alonso2009h}
Alonso, S., Cabrerizo, F.J., Herrera-Viedma, E.\ \& 
Herrera, F. (2009). $h$-Index: a review focused in its variants, computation and standardization for different scientific fields. 
\emph{Journal of Informetrics}, 
{\bf 3}\myp(4), 273--289. \href{https://doi.org/10.1016/j.joi.2009.04.001}{https://\allowbreak doi.\allowbreak org/\allowbreak 10.\allowbreak 1016/\allowbreak j.\allowbreak joi.\allowbreak 2009.\allowbreak 04.\allowbreak 001}


\bibitem[Bogachev(2006)]{Bogachev2006}
Bogachev, L. (2006). Limit laws for norms of IID samples with Weibull tails. 
\emph{Journal of Theoretical Probability}, 
{\bf 19}\myp(4), 849--873. 
\href{https://doi.org/10.1007/s10959-006-0036-z}{https://\allowbreak doi.\allowbreak org/\allowbreak 10.\allowbreak 1007/\allowbreak s10959-006-0036-z}


\bibitem[Costas \& Bordons(2007)]{Costas2007}
{Costas, R.\ \& Bordons, M. (2007). The h-index: advantages, limitations and its relation with other bibliometric indicators at the micro level. 
\emph{Journal of Informetrics}, 
{\bf 1}\myp(3), 193--203. \href{https://doi.org/10.1016/j.joi.2007.02.001}{https:\allowbreak//\allowbreak doi.\allowbreak org/\allowbreak 10.\allowbreak 1016/\allowbreak j.\allowbreak joi.\allowbreak 2007.\allowbreak 02.\allowbreak 001}}

\bibitem[DORA(2020)]{DORA}
{DORA. (2020). San Francisco Declaration of Research Assessment. American Society for Cell Biology (ASCB). (Online) \href{https://sfdora.org/read}{https:\allowbreak //\allowbreak sfdora.\allowbreak org/\allowbreak read}}



\bibitem[Egghe(2006a)]{egghe2006improvement}
Egghe, L. (2006a). An improvement of the h-index: the g-index. 
\emph{ISSI Newsletter}, 
{\bf 2}\myp(1), 8--9. \href{https://www.issi-society.org/media/1182/newsletter05.pdf}{https:\allowbreak //\allowbreak www.\allowbreak issi-society.\allowbreak org/\allowbreak media/\allowbreak 1182/\allowbreak newsletter05.pdf}


\bibitem[Egghe(2006b)]{egghe2006theory}
Egghe, L. (2006b). Theory and practise of the $g$-index. \emph{Scientometrics},\ {\bf 69}\myp(1), 131--152. \href{https://doi.org/10.1007/s11192-006-0144-7}{https://\allowbreak doi.\allowbreak org/\allowbreak 10.\allowbreak 1007/\allowbreak s11192-\allowbreak 006-\allowbreak 0144-7}

\bibitem[Egghe \& Rousseau(2006)]{er2006}
Egghe, L.\ \& Rousseau, R. (2006). An informetric model for the Hirsch-index. \emph{Scientometrics}, {\bf 69}\myp(1), 121--129. \href{https://doi.org/10.1007/s11192-006-0143-8}{https://\allowbreak doi.\allowbreak org/\allowbreak 10.\allowbreak 1007/\allowbreak s11192-006-0143-8}

\bibitem[Egghe \& Rousseau(2008)]{egghe2008h}
Egghe, L.\ \& Rousseau, R. (2008). An $h$-index weighted by citation impact. 
\emph{Information Processing \& Management}, 
{\bf 44}\myp(2), 770--780. \href{https://doi.org/10.1016/j.ipm.2007.05.003}{https:\allowbreak //\allowbreak doi.\allowbreak org/\allowbreak 10.\allowbreak 1016/\allowbreak j.\allowbreak ipm.\allowbreak 2007.\allowbreak 05.003}

\bibitem[Egghe et al.(2019)]{Egghe2019}
Egghe, L., Fassin, Y.\ \& Rousseau, R. (2019). Equalities between h-type indices and definitions of rational h-type indicators. 
\emph{Journal of Data and Information Science},
{\bf 4}\myp(1), 22--31. \href{https://doi.org/10.2478/jdis-2019-0002}{https://\allowbreak doi.\allowbreak org/\allowbreak 10.\allowbreak 2478/\allowbreak jdis-2019-0002}


\bibitem[Guns \& Rousseau(2009)]{guns2009real}
Guns, R.\ \& Rousseau, R. (2009). Real and rational variants of the $h$-index and the $g$-index. 
{\em Journal of Informetrics},
{\bf 3}\myp(1), 64--71. \href{https://doi.org/10.1016/j.joi.2008.11.004}{https://\allowbreak doi.\allowbreak org/\allowbreak 10.\allowbreak 1016/\allowbreak j.\allowbreak joi.\allowbreak 2008.\allowbreak 11.\allowbreak 004}


\bibitem[Hicks et al.(2015)]
{Hicks2015}
Hicks, D., Wouters, P., Waltman, L., de Rijcke, S.\ \& Rafols, I. (2015). Bibliometrics: the Leiden Manifesto for research metrics. \emph{Nature}, {\bf 520}\myp(7548), 429--431. \href{https://doi.org/10.1038/520429a}{https://\allowbreak doi.\allowbreak org/\allowbreak 10.\allowbreak 1038\allowbreak /\allowbreak 52\allowbreak 0429a}

\bibitem[Hirsch(2005)]{hirsch2005index}
Hirsch, J.E. (2005). An index to quantify an individual's scientific research output. \emph{Proceedings of the National Academy of Sciences of the United States of America\/}, 
{\bf 102}\allowbreak \myp(46), \allowbreak 16569--16572. 
\href{https://doi.org/10.1073/pnas.0507655102}{https://\allowbreak doi.\allowbreak org/\allowbreak 10.\allowbreak 1073/\allowbreak pnas.\allowbreak 05\allowbreak 07\allowbreak 6\allowbreak 5\allowbreak 51\allowbreak 02}

\bibitem[Hirsch(2007)]{hirsch2007index}
Hirsch, J.E. (2007). Does the $h$ index have predictive power? \emph{Proceedings of the National Academy of Sciences of the United States of America}, 
{\bf 104}\myp(49), 19193--19198. \href{https://doi.org/10.1073/pnas.0707962104}{https:\allowbreak //\allowbreak doi.\allowbreak org/\allowbreak 10.\allowbreak 1073/\allowbreak pnas.\allowbreak 0707962104}

\bibitem[Hirsch(2019)]{hirsch2019index}
Hirsch, J.E. (2019). $h_\alpha$:\ An index to quantify an individual’s scientific leadership. \emph{Scientometrics}, {\bf 118}\myp(2), 673--686. \href{https://doi.org/10.1007/s11192-018-2994-1}{https://\allowbreak doi.\allowbreak org/\allowbreak 10.\allowbreak 1007/\allowbreak s11192-\allowbreak 018-\allowbreak 2994-1}

\bibitem[Marshall et al.(2011)]{Marshall}
Marshall, A.W., Olkin, I.\ \& Arnold, B.C. (2011).
\emph{Inequalities: Theory of Majorization and Its Applications}, 2nd ed. Springer Series in Statistics. New York: Springer. \href{https://doi.org/10.1007/978-0-387-68276-1}{https:\allowbreak //\allowbreak doi.\allowbreak org/\allowbreak 10.\allowbreak 1007/\allowbreak 978-\allowbreak 0-\allowbreak 387-\allowbreak 68276-1}

\bibitem[Nuermaimaiti(2023)]{Nuer2023}
Nuermaimaiti, R. (2023).
\emph{Statistical Models for Frequency Distributions of Count Data with Applications to Scientometrics}. 
PhD thesis, University of Leeds, White Rose eTheses Online. 
\href{https://etheses.whiterose.ac.uk/id/oai_id/oai:etheses.whiterose.ac.uk:33446}{https://\allowbreak etheses.\allowbreak whiterose.\allowbreak ac.\allowbreak uk/\allowbreak id/\allowbreak oai\_\allowbreak id/\allowbreak oai:\allowbreak etheses.\allowbreak whiterose.\allowbreak ac.\allowbreak uk:33446}

\bibitem[Nuermaimaiti et al.(2021)]{nuerISSI2021}
Nuermaimaiti, R., Bogachev, L.V.\ \& Voss, J. (2021). A generalized power law
model of citations.  In: Gl\"anzel, W., Heeffer, S., Chi, P.-S.\ \&  Rousseau, R.\ (Eds). \emph{Proceedings, 18th International Conference on Scientometrics \& Informetrics (ISSI\,2021)}
(pp.\ 843--848). International Society for Sciento\-metrics and Informetrics (I.S.S.I.).
\href{https://www.issi-society.org/proceedings/issi_2021/Proceedings%20ISSI%202021.pdf}{https://\allowbreak www.\allowbreak issi-society.\allowbreak org/\allowbreak pro\allowbreak ceed\allowbreak ings/\allowbreak issi\_2021/\allowbreak Proceedings\allowbreak \%20\allowbreak ISSI\allowbreak \%20\allowbreak 2021.\allowbreak pdf}

\bibitem[Schlather(2001)]{Schlather}
Schlather, M. (2001). Limit distributions of norms of vectors of positive i.i.d.\ random variables.
\emph{Annals of Probability},
{\bf 29}\myp(2), 862--881. \href{https://doi.org/10.1214/aop/1008956695}{https://\allowbreak doi.\allowbreak org/\allowbreak 10.\allowbreak 1214/\allowbreak aop/\allowbreak 100\allowbreak 895\allowbreak 6695}

\bibitem[SCI$^2$S(2025)]{web-h}
SCI$^2$S. (2025). Soft Computing and Intelligent Information Systems. (Online) 
\emph{h-index and Variants}. 
University of Granada.  \href{https://sci2s.ugr.es/hindex}{https://\allowbreak sci2s.\allowbreak ugr.\allowbreak es/\allowbreak hindex}.

\bibitem[Thelwall(2025)]
{Thelwall2025}
Thelwall, M. (2025). \emph{Quantitative Methods in Research Evaluation: Citation Indicators, Altmetrics, and Artificial Intelligence}. Preprint, arXiv:2407.00135 [cs.DL]. \href{https://arxiv.org/abs/2407.00135}{https:\allowbreak //\allowbreak arxiv.\allowbreak org/\allowbreak abs/\allowbreak 2407.\allowbreak 00135}


\bibitem[Waltman(2016)]{waltman2016}
Waltman, L. (2016). A review of the literature on citation impact indicators.
\emph{Journal of Informetrics},
{\bf 10}\myp(2), 365--391. \href{https://doi.org/10.1016/j.joi.2016.02.007}{https://\allowbreak doi.\allowbreak org/\allowbreak 10.\allowbreak 1016/\allowbreak j.\allowbreak joi.\allowbreak 2016.\allowbreak 02.\allowbreak 007}

\bibitem[WoS(2025)]
{WoS}
WoS. (2025). Web of Science. (Online)
\href{https://clarivate.com/academia-government/scientific-and-academic-research/research-discovery-and-referencing/web-of-science/}{https:\allowbreak //\allowbreak clarivate.\allowbreak com/\allowbreak academia-\allowbreak government/ \allowbreak scientific-\allowbreak and-\allowbreak academic-\allowbreak research/\allowbreak research-\allowbreak discovery-\allowbreak and-\allowbreak referencing/\allowbreak web-\allowbreak of-\allowbreak science}). 

\bibitem[Woeginger(2008a)]{woeginger2008axH}
Woeginger, G.J. (2008a). An axiomatic characterization of the Hirsch-index.
\emph{Mathematical Social Sciences}, 
{\bf 56}\myp(2), 224--232. \href{https://doi.org/10.1016/j.mathsocsci.2008.03.001}{https://\allowbreak doi.\allowbreak org/\allowbreak 10.\allowbreak 1016/\allowbreak j.\allowbreak mathsocsci.\allowbreak 2008.\allowbreak 03.\allowbreak 001}


\bibitem[Woeginger(2008b)]{woeginger2008axiomatic}
Woeginger, G.J. (2008b). An axiomatic analysis of Egghe’s $g$-index.
\emph{Journal of Informetrics},
{\bf 2}\myp(4), 364--368. \href{https://doi.org/10.1016/j.joi.2008.05.002}{https://\allowbreak doi.\allowbreak org/\allowbreak 10.\allowbreak 1016/\allowbreak j.\allowbreak joi.\allowbreak 2008.\allowbreak 05.\allowbreak 002}




\end{thebibliography}

\end{document}